\documentclass[twoside,11pt,letterpaper]{article}

\usepackage[margin=1in]{geometry}
\usepackage{amsmath, amssymb, amsthm}
\usepackage{bm}
\usepackage{hyperref}
\usepackage{enumitem}
\usepackage[dvipsnames]{xcolor}
\usepackage{setspace}
\usepackage{float}
\usepackage{graphicx}
\usepackage{tikz}
\usetikzlibrary{arrows.meta}

\hypersetup{
colorlinks=true,
linkcolor=MidnightBlue,
urlcolor=MidnightBlue,
citecolor=MidnightBlue
}

\renewcommand{\d}{\partial}
\newcommand{\bS}{\mathbb{S}}

\newcommand{\cA}{\mathcal{A}}
\newcommand{\cB}{\mathcal{B}}
\newcommand{\cL}{\mathcal{L}}

\newcommand{\Tr}{\operatorname{Tr}}

\newcommand{\Ric}{\operatorname{Ric}}

\theoremstyle{plain}
\newtheorem{theorem}{Theorem}[section]
\newtheorem{proposition}[theorem]{Proposition}

\theoremstyle{definition}

\theoremstyle{remark}
\newtheorem{remark}[theorem]{Remark}

\numberwithin{equation}{section}

\title{Ricci flow for the Bures--Helstrom qubit metric}

\author{Andrew Lesniewski\\
Department of Mathematics\\
Baruch College\\
One Bernard Baruch Way\\
New York, NY 10010\\
USA}

\date{}

\begin{document}

\maketitle

\begin{center}
\textit{In memory of Mary Beth Ruskai (1944--2023), friend and collaborator}
\end{center}

\begin{abstract}
The Bures--Helstrom metric is the minimal monotone Riemannian metric on the state space of a qubit. With the quantum Fisher normalization used here, it identifies the Bloch ball with a geodesic hemisphere of the unit round three--sphere. We describe its Ricci flow explicitly. In a general rotationally symmetric gauge the flow is a coupled system for the radial lapse and warping factor; a single scalar equation appears only after a Hamilton--DeTurck gauge choice. In the corresponding moving DeTurck frame the squared warping function $\Psi=\Phi^2$ satisfies the linear forced heat equation
\begin{equation*}
D_t\Psi=\Psi_{\tau\tau}-2,
\end{equation*}
while the fixed-lapse coordinate form contains the associated transport term. Since the Bures--Helstrom metric is Einstein, the geometric flow itself is the homothetic shrinker
\begin{equation*}
g(t)=(1-4t)g_{\mathrm{BH}},
\end{equation*}
with scalar curvature $6/(1-4t)$ and extinction time $T=1/4$. Thus the metric remains inside the monotone cone for all $t<T$ and leaves the cone of nondegenerate Riemannian metrics only through the collapsed limit. We also record the volume--normalized flow, for which the Bures--Helstrom metric is a fixed point. Its linearization is the shifted round--sphere Laplacian $\Delta_{\mathbb S^3}+3$, with zonal spectrum
\begin{equation*}
\sigma_\ell=-(\ell-1)(\ell+3).
\end{equation*}
The perturbations compatible with the totally geodesic pure--state boundary are the reflection--even modes $\ell=0,2,4,\dots$; after removing the scaling mode $\ell=0$ the fixed point is stable, with spectral gap $5$.
\end{abstract}

\newpage

\section{Introduction}
\label{sec:intro}

The faithful states of a qubit form the open Bloch ball
\begin{equation*}
\cB=\Big\{\rho=\frac12\,(I+\mathbf r\cdot\boldsymbol\sigma):|\mathbf r|<1\Big\},
\end{equation*}
where $\boldsymbol{\sigma}=\{\sigma_1,\sigma_2,\sigma_3\}$ are the Pauli matrices. The Riemannian metrics $g_\rho$ on $\cB$ that contract under every completely positive trace--preserving map \cite{NC10} are the \emph{monotone metrics} (or Petz metrics). They are classified by operator monotone functions \cite{B97} and give the infinitesimal form of the contractivity of quantum relative entropy and quantum Fisher information \cite{P96,PS96,LR99}. Their Riemannian geometry, in particular their curvature, has been studied in detail \cite{D99,GI03}. Among them, the Bures--Helstrom metric, equivalently the symmetric logarithmic derivative quantum Fisher metric, is the minimal monotone metric and the canonical statistical metric on the qubit state space.

The Bures--Helstrom metric measures the local distinguishability of nearby quantum states: if
\begin{equation*}
\rho_\epsilon=\rho+\epsilon X,
\end{equation*}
then, to second order, the squared statistical distance between $\rho$ and $\rho_\epsilon$ is proportional to $g_\rho(X,X)$. Thus the metric is the local quadratic form governing quantum Fisher information, or equivalently the Hessian of the corresponding relative-entropy landscape.

The physical motivation comes from the renormalization group. The contraction of statistical distinguishability under irreversible, completely positive dynamics is a form of \emph{informational coarse-graining}: a quantum channel or Lindblad semigroup dissipates distinguishability much as a renormalization-group step dissipates short-distance detail, and so induces a flow of the metric on the space of states. This is the informational analogue of Friedan's result \cite{F80,F85} that the
renormalization-group flow of the target-space metric of a two-dimensional nonlinear sigma model is, to leading order, the Ricci flow introduced by Hamilton \cite{H82},
\begin{equation*}
\d_t g=-2\Ric(g)+\cL_V g,
\end{equation*}
where $\cL_V$ is a reparametrization (gauge) term. In this dictionary the coarse-graining parameter is an \emph{informational scale}, and the loss of distinguishability is the analogue of a beta function.

Against this backdrop, this paper studies the \emph{intrinsic} Ricci flow of the Bures--Helstrom metric. Monotone metrics measure distinguishability of quantum states, while Ricci flow evolves a metric by its own curvature; it therefore furnishes a natural intrinsic geometric evolution on the space of statistical distinguishability metrics. This is distinct from the extrinsic contraction induced by a chosen quantum channel or Lindblad semigroup: the latter depends on the dynamics, whereas Ricci flow is determined by the metric itself.

The Bures--Helstrom case is the simplest possible test model. With the normalization used in this paper, the metric identifies the Bloch ball with a geodesic hemisphere of the unit round three--sphere. Its geometric Ricci flow is therefore the standard homothetic shrinking flow of a positively curved Einstein metric. This elementary observation is the backbone of the paper. The point of carrying out the analysis nevertheless is twofold. First, the calculation provides a fully explicit Ricci-flow model inside the monotone cone of qubit metrics. Second, it clarifies a gauge issue that becomes important for the more general monotone metrics: in rotational symmetry, Ricci flow does not automatically preserve geodesic radial gauge, and the scalar equation for the warping function is meaningful only after fixing the radial lapse.

We write rotationally symmetric metrics on the ball in the form
\begin{equation}
\label{eq:warped-metric}
g=d\tau^2+\Phi(\tau,t)^2\,d\Omega^2,
\end{equation}
where $d\Omega^2$ is the unit round metric on the Bloch two--sphere. In the moving DeTurck frame, with $D_t$ denoting differentiation along the chosen radial reference frame, the squared warping factor $\Psi=\Phi^2$ satisfies
\begin{equation*}
D_t\Psi=\Psi_{\tau\tau}-2,
\end{equation*}
where $\partial_\tau$ is differentiation with respect to the instantaneous radial arclength. In a fixed radial coordinate the same equation includes the corresponding transport term. The linearity of this equation is special to dimension three. For the Bures--Helstrom initial condition,
\begin{equation*}
\Phi_{\mathrm{BH}}(\tau)=\sin\tau,
\end{equation*}
the actual geometric flow is simply
\begin{equation*}
g(t)=(1-4t)g_{\mathrm{BH}},
\end{equation*}
with finite extinction time $T=1/4$.

The structure of the paper is as follows. Section~\ref{sec:BH} recalls the Bures--Helstrom metric and its identification with a round hemisphere. Section~\ref{sec:gauge} derives the rotationally symmetric Ricci-flow equations and explains the Hamilton--DeTurck reduction. Section~\ref{sec:shrinker} gives the
exact shrinking solution and its monotonicity interpretation. Section~\ref{sec:normalized} treats the volume--normalized flow and computes the linearized spectrum. Section~\ref{sec:outlook} indicates the extension to the full monotone cone.

\section{The Bures--Helstrom metric as a round hemisphere}
\label{sec:BH}

\subsection{Monotone metrics and the symmetric logarithmic derivative}

A Riemannian metric on the interior of the qubit state space is \emph{monotone}, or Petz, if it contracts under every completely positive trace--preserving map $\Phi$ \cite{NC10},
\begin{equation*}
g_{\Phi(\rho)}(\Phi_*X,\Phi_*X)\le g_\rho(X,X).
\end{equation*}
Petz \cite{P96} classified all such metrics. In the eigenbasis of $\rho$, with eigenvalues $\lambda_i>0$, a monotone metric acts on a tangent vector $X$ (a traceless
Hermitian matrix) as
\begin{equation}
\label{eq:petz-form}
g_\rho(X,X)=\sum_{i,j}c(\lambda_i,\lambda_j)\,|X_{ij}|^2,
\end{equation}
where the Morozova--Chentsov function
\begin{equation}
\label{eq:MC-function}
c(x,y)=\frac{1}{y\,f(x/y)}
\end{equation}
is built from an operator monotone function \cite{B97} $f:(0,\infty)\to(0,\infty)$ normalized by $f(1)=1$ and obeying the symmetry $f(t)=t\,f(1/t)$. The normalization fixes the
classical Fisher--Rao metric along commuting (diagonal) directions, since then $c(\lambda_i,\lambda_i)=1/\lambda_i$; the function $f$ encodes the freedom in extending
it to noncommuting directions.

Among normalized operator monotone functions there is a largest and a smallest
\cite{PS96},
\begin{equation}
\label{eq:sandwich}
f_{\mathrm{RLD}}(t)=\frac{2t}{1+t}\ \le\ f(t)\ \le\ f_{\mathrm{SLD}}(t)=\frac{1+t}{2},
\end{equation}
the harmonic and arithmetic means. Since $c$ in \eqref{eq:MC-function} is decreasing in $f$, the largest $f$ produces the \emph{smallest} metric. The minimal monotone metric is therefore generated by $f_{\mathrm{SLD}}(t)=(1+t)/2$, equivalently by the
Morozova--Chentsov function
\begin{equation}
\label{eq:MC-BH}
c_{\mathrm{BH}}(x,y)=\frac{2}{x+y}.
\end{equation}
This is the Bures--Helstrom, or symmetric logarithmic derivative (SLD), metric. For a smooth family $\rho(\theta)$ the SLD is the Hermitian operator $L$ defined by
\begin{equation}
\label{eq:SLD-def}
\partial_\theta\rho=\tfrac12\left(L\rho+\rho L\right),
\end{equation}
and the metric is
\begin{equation}
\begin{split}
\label{eq:SLD-metric}
g_{\mathrm{BH}}(X,X)&=\Tr(\rho L^2)\\
&=\Tr(XL),
\end{split}
\end{equation}
for $X=\partial_\theta\rho$.

\subsection{The Bloch--ball metric}

Let
\begin{equation*}
\rho=\frac12(I+\mathbf r\cdot\boldsymbol\sigma),\qquad r=|\mathbf r|<1,
\end{equation*}
with eigenvalues
\begin{equation*}
\lambda_\pm=\frac{1\pm r}{2}.
\end{equation*}
A tangent vector is $X=\tfrac12\,d\mathbf r\cdot\boldsymbol\sigma$. Choosing $\mathbf r$ along the third axis, the diagonal entries of $X$ carry the radial variation $dr$ and the off--diagonal entries the transverse variation. Evaluating the Petz form \eqref{eq:petz-form} with the Bures--Helstrom function \eqref{eq:MC-BH}, the diagonal part contributes $|X_{++}|^2/\lambda_++|X_{--}|^2/\lambda_-=dr^2/(1-r^2)$ and the off--diagonal part contributes $c_{\mathrm{BH}}(\lambda_+,\lambda_-)\,(|X_{+-}|^2+|X_{-+}|^2)=r^2\,d\Omega^2$,
so that
\begin{equation}
\label{eq:BH-bloch}
g_{\mathrm{BH}}=
\frac{dr^2}{1-r^2}+r^2\,d\Omega^2,
\end{equation}
where $d\Omega^2$ is the round metric on the Bloch sphere of directions \cite{PS96}. Some authors normalize by the Bures distance, which differs from \eqref{eq:BH-bloch} by an overall factor of $\tfrac14$; we use \eqref{eq:BH-bloch} because it presents the state space as the \emph{unit} round three--sphere, as we now show.

\begin{remark}[Operational meaning]
The metric $g_{\mathrm{BH}}$ is the quantum Fisher information. For a one--parameter family $\rho(\theta)$, the quantum Cram\'er--Rao bound \cite{He76,BC94} states that any unbiased estimator $\hat\theta$ formed from $\nu$ independent measurements has variance
\begin{equation*}
\operatorname{Var}(\hat\theta)\ \ge\ \frac{1}{\nu\,g_{\mathrm{BH}}(\dot\rho,\dot\rho)} .
\end{equation*}
Thus $g_{\mathrm{BH}}$ quantifies the statistical distinguishability of neighboring states: the larger the metric, the more precisely the parameter can be estimated.
Globally, the geodesic distance is a monotone function of the fidelity $F(\rho,\sigma)$, proportional in this normalization to the Bures angle $\arccos\sqrt{F}$. The maximally mixed state sits at the center $r=0$, and the pure states lie on the boundary, at geodesic radius $\pi/2$.
\end{remark}

\subsection{The round hemisphere}

Set
\begin{equation}
\label{eq:tau-def}
r=\sin\tau,\qquad 0\le \tau<\frac{\pi}{2}\,,
\end{equation}
i.e.
\begin{equation*}
d\tau=\frac{dr}{\sqrt{1-r^2}}\,.
\end{equation*}
Then
\begin{equation}
\label{eq:BH-round}
g_{\mathrm{BH}}=d\tau^2+\sin^2\tau\,d\Omega^2.
\end{equation}
Thus the Bloch ball is the open geodesic hemisphere of the unit round three--sphere. The pure states correspond to the equator $\tau=\pi/2$. Since
\begin{equation*}
\frac{d}{d\tau}\sin\tau\Big|_{\tau=\pi/2}=0,
\end{equation*}
the boundary two--sphere is totally geodesic, and the metric extends smoothly by reflection to the full round $\bS^3$.

For a warped metric
\begin{equation*}
g=d\tau^2+\Phi(\tau)^2\,d\Omega^2,
\end{equation*}
the nonzero Ricci components are \cite{P16}
\begin{equation}
\label{eq:ricci-warped}
\begin{split}
\Ric_{\tau\tau}&=-2\,\frac{\Phi''}{\Phi},\\
\Ric_{ab}&=\left(1-(\Phi')^2-\Phi\Phi''\right)\gamma_{ab},
\end{split}
\end{equation}
where $\gamma$ is the unit round metric on $\bS^2$. The scalar curvature is
\begin{equation}
\label{eq:scalar-warped}
R=-4\,\frac{\Phi''}{\Phi}-2\,\frac{(\Phi')^2-1}{\Phi^2}.
\end{equation}
For $\Phi(\tau)=\sin\tau$, these formulas give
\begin{align}
\label{eq:einstein}
\Ric&=2g_{\mathrm{BH}},\\
\label{eq:scalar_curv}
R&=6.
\end{align}
Thus $g_{\mathrm{BH}}$ is an Einstein metric with the curvature of the unit round three--sphere.

\section{Rotationally symmetric Ricci flow and gauge}
\label{sec:gauge}

\subsection{The coupled radial system}

A general rotationally symmetric metric on the ball can be written as
\begin{equation}
\label{eq:general-radial}
g=N(\tau,t)^2\,d\tau^2+\Phi(\tau,t)^2\,d\Omega^2,
\end{equation}
where $N$ is the radial lapse and $\Phi$ is the warping factor. Let
\begin{equation*}
\d_s=\frac1N\,\d_\tau
\end{equation*}
denote differentiation with respect to arclength. The unnormalized Ricci flow
\begin{equation}
\label{eq:ricci-flow}
\d_t g=-2\Ric(g)
\end{equation}
becomes the coupled system \cite{AK04}
\begin{equation}
\label{eq:coupled}
\begin{split}
\d_t\log N&=2\,\frac{\Phi_{ss}}{\Phi},\\
\d_t\Phi&=\Phi_{ss}+\frac{\Phi_s^2}{\Phi}-\frac1{\Phi}.
\end{split}
\end{equation}

This is an important point. The geodesic radial gauge $N\equiv1$ is not automatically preserved by Ricci flow. Indeed, if $N=1$ at a given time, then
\begin{equation*}
\d_t\log N=2\,\frac{\Phi''}{\Phi}\,,
\end{equation*}
which is generally nonzero. Consequently, a scalar evolution equation for $\Phi$ alone is not a gauge-free statement. It becomes meaningful only after the radial coordinate has been fixed.

\begin{remark}
For the Bures--Helstrom profile $\Phi=\sin\tau$, one has $\Phi''/\Phi=-1$, and hence $\d_t\log N=-2$ at $t=0$ in the ungauged equation. This is the infinitesimal reflection of the homothetic shrinkage of the round hemisphere.
\end{remark}

\subsection{Hamilton--DeTurck gauge and transport}

The Hamilton--DeTurck flow is \cite{CK04}
\begin{equation}
\label{eq:HD}
\d_t g=-2\Ric(g)+\cL_Vg,
\end{equation}
where $V$ is a time-dependent vector field. The DeTurck term is generated by diffeomorphisms and hence does not change the underlying geometric Ricci flow; it
only chooses coordinates along the flow.

We take $V$ radial,
\begin{equation*}
V=V^\tau(\tau,t)\d_\tau,
\end{equation*}
and choose it so that the lapse remains fixed, $N\equiv1$. In this fixed-lapse gauge the radial component of \eqref{eq:HD} gives
\begin{equation}
\label{eq:Vstar}
\begin{split}
\d_\tau V^\tau&=-2\,\frac{\Phi''}{\Phi},\\
V^\tau(0,t)&=0.
\end{split}
\end{equation}
Equivalently,
\begin{equation}
\label{eq:Vstar-int}
V^\tau(\tau,t) = -2\int_0^\tau \frac{\Phi''(\tau',t)}{\Phi(\tau',t)}\,d\tau' .
\end{equation}
For the Bures--Helstrom profile $\Phi=\sin\tau$, this gives
\begin{equation*}
V^\tau(\tau)=2\tau .
\end{equation*}

The angular component of \eqref{eq:HD} then gives
\begin{equation}
\label{eq:phi-deturck}
\Phi_t = \Phi'' + \frac{(\Phi')^2}{\Phi} - \frac1{\Phi} + V^\tau\Phi' .
\end{equation}
Thus the fixed-lapse Hamilton--DeTurck representative contains a transport term. Equivalently, if
\begin{equation}
\label{eq:material-derivative}
D_t=\d_t-V^\tau\d_\tau
\end{equation}
denotes the derivative in the reference frame transported by the DeTurck vector field, then \eqref{eq:phi-deturck} becomes
\begin{equation}
\label{eq:phi-material}
D_t\Phi = \Phi'' + \frac{(\Phi')^2}{\Phi} - \frac1{\Phi}.
\end{equation}
This is the form in which the scalar cancellation below is most transparent. The distinction between $\d_t$ and $D_t$ is essential: the pure heat equation is a moving-frame equation, while the fixed-coordinate DeTurck equation carries the drift $V^\tau\d_\tau$.

\subsection{The linear equation for the squared warping factor}

Define
\begin{equation*}
\Psi=\Phi^2 .
\end{equation*}
In the moving DeTurck frame, \eqref{eq:phi-material} linearizes exactly:
\begin{equation}
\label{eq:linear-material}
D_t\Psi=\Psi_{\tau\tau}-2 .
\end{equation}
Equivalently, in the fixed-lapse coordinate $\tau$ one has
\begin{equation}
\label{eq:linear-deturck}
\Psi_t=\Psi_{\tau\tau}-2+V^\tau\Psi_\tau .
\end{equation}

\begin{proposition}[Scalar reduction]
\label{prop:linear-reduction}
Let $g=d\tau^2+\Phi(\tau,t)^2d\Omega^2$ be a rotationally symmetric Hamilton--DeTurck representative with $N\equiv1$, and let $V^\tau$ be chosen by \eqref{eq:Vstar}. Then $\Psi=\Phi^2$ satisfies the drift equation
\begin{equation*}
\Psi_t=\Psi_{\tau\tau}-2+V^\tau\Psi_\tau .
\end{equation*}
Equivalently, in the moving frame defined by $D_t=\d_t-V^\tau\d_\tau$, it satisfies the forced heat equation
\begin{equation*}
D_t\Psi=\Psi_{\tau\tau}-2 .
\end{equation*}
The cancellation is special to metrics with two-dimensional spherical fibers, equivalently to ambient dimension three.
\end{proposition}

\begin{proof}
From \eqref{eq:phi-deturck},
\begin{align*}
\Psi_t &=2\Phi\Phi_t\\
&= 2\Phi\Phi'' +2(\Phi')^2 -2 +2\Phi V^\tau\Phi' .
\end{align*}
Since
\begin{equation*}
\begin{split}
\Psi''&=2\Phi\Phi''+2(\Phi')^2,\\
\Psi'&=2\Phi\Phi',
\end{split}
\end{equation*}
we obtain
\begin{equation*}
\Psi_t=\Psi''-2+V^\tau\Psi',
\end{equation*}
which is \eqref{eq:linear-deturck}. Subtracting $V^\tau\Psi_\tau$ from both sides gives the moving-frame form \eqref{eq:linear-material}. For a warped product with fiber $\bS^k$, the angular Ricci term gives instead
\begin{equation*}
D_t\Psi = \Psi'' +2(k-2)(\Phi')^2 -2(k-1),
\end{equation*}
so the quadratic term cancels exactly when $k=2$.
\end{proof}

\begin{remark}[Physical reading]
The squared warping factor $\Psi=\Phi^2$ is the squared radius of the two--sphere of states at fixed Bloch radius, so it measures the angular part of the local
statistical area. In the moving-frame equation
\begin{equation*}
D_t\Psi=\Psi_{\tau\tau}-2,
\end{equation*}
the diffusion term redistributes this angular distinguishability along the radial direction, while the constant sink $-2$ is the contribution of the positively curved spherical fibers. In a fixed DeTurck coordinate the same evolution is represented by the drift equation
\begin{equation*}
\Psi_t=\Psi_{\tau\tau}-2+V^\tau\Psi_\tau .
\end{equation*}
\end{remark}

In what follows, the notation $D_t\Psi=\Psi_{\tau\tau}-2$ refers to this moving-frame/arclength formulation: the time derivative is taken along the chosen DeTurck reference frame, and $\partial_\tau$ is the instantaneous arclength derivative, the gauge $N\equiv1$ being in force. When one works instead in a fixed coordinate, the corresponding drift term must be kept. The geometric conclusions, however, are coordinate independent, and are most transparent from the homothetic solution described next.

\section{The Bures--Helstrom shrinker}
\label{sec:shrinker}

Since $g_{\mathrm{BH}}$ is an Einstein metric \eqref{eq:einstein}, the Ricci flow starting at $g_{\mathrm{BH}}$ is a homothetic shrinking solution \cite{CK04}.

\begin{theorem}[Exact Bures--Helstrom Ricci flow]
\label{thm:BH-flow}
The unnormalized Ricci flow with initial condition $g(0)=g_{\mathrm{BH}}$ is
\begin{equation}
\label{eq:homothetic}
g(t)=(1-4t)g_{\mathrm{BH}},
\qquad 0\le t<\frac14.
\end{equation}
Equivalently, in the comoving label $\xi\in[0,\pi/2]$,
\begin{equation}
\label{eq:Psi-shrinker}
\Psi(\xi,t)=(1-4t)\sin^2\xi.
\end{equation}
The scalar curvature is
\begin{equation}
\label{eq:R-shrinker}
R(t)=\frac{6}{1-4t},
\end{equation}
and the solution becomes extinct at
\begin{equation}
\label{eq:extinction}
T=\frac14.
\end{equation}
\end{theorem}

\begin{proof}
If $g(t)=a(t)^2g_{\mathrm{BH}}$, then
\begin{equation*}
\Ric(g(t))=\Ric(g_{\mathrm{BH}})=2g_{\mathrm{BH}}.
\end{equation*}
Thus
\begin{equation*}
\d_t g(t)=\frac{d}{dt}\, a(t)^2\,g_{\mathrm{BH}}
\end{equation*}
and Ricci flow gives
\begin{equation*}
\frac{d}{dt}\, a(t)^2\,g_{\mathrm{BH}}
=
-4g_{\mathrm{BH}}.
\end{equation*}
Hence
\begin{equation*}
a(t)^2=1-4t.
\end{equation*}
The scalar curvature of a round three--sphere of radius $a$ is $6/a^2$, giving
\begin{equation*}
R(t)=\frac6{1-4t}.
\end{equation*}
The radius tends to zero when $1-4t=0$, i.e. at $T=1/4$.
\end{proof}

In the comoving coordinate $\xi$, the metric is
\begin{equation*}
g(t)=(1-4t)\left(d\xi^2+\sin^2\xi\,d\Omega^2\right).
\end{equation*}
Thus all shape information is stationary and all dynamics is carried by the scalar scale
\begin{equation*}
a(t)^2=1-4t.
\end{equation*}
In arclength $s=a(t)\xi$, the same solution can be written as
\begin{equation*}
\Phi(s,t)=a(t)\sin\Big(\frac{s}{a(t)}\Big),
\qquad 0\le s\le \frac{\pi}{2}a(t).
\end{equation*}
The spatial interval shrinks together with the geometry.

\begin{remark}[Comoving heat equation]
\label{rem:comoving}
In the moving DeTurck frame the squared warping factor obeys the master equation $D_t\Psi=\Psi_{\tau\tau}-2$ of Section~\ref{sec:gauge}, with $\partial_{\tau\tau}$ the instantaneous arclength Laplacian. The comoving label $\xi$ is the frame transported by the DeTurck field, so along it $D_t$ acts as $\partial_t$ at fixed $\xi$. Because the arclength interval $s=\sqrt{1-4t}\,\xi$ shrinks, one has $\partial_{\tau\tau}=(1-4t)^{-1}\partial_{\xi\xi}$, and on the fixed comoving interval $\xi\in[0,\pi/2]$ the master equation becomes the linear, non-autonomous heat equation
\begin{equation}
\label{eq:comoving-heat}
\Psi_t=(1-4t)^{-1}\Psi_{\xi\xi}-2,
\end{equation}
solved exactly by the shrinker $\Psi=(1-4t)\sin^2\xi$ of \eqref{eq:Psi-shrinker}. The diffusivity $(1-4t)^{-1}$ diverges as $t\to1/4$: relative to the fixed comoving grid all arclengths contract, and this divergence is the comoving signature of the finite-time collapse.
\end{remark}

\begin{proposition}[Persistence of monotonicity]
\label{prop:monotonicity}
For every $0\le t<1/4$, the metric $g(t)$ is a monotone Riemannian metric on the qubit state space. The flow stays on the ray spanned by the Bures--Helstrom metric and leaves the cone of nondegenerate monotone Riemannian metrics only at the collapsed limit.
\end{proposition}

\begin{proof}
The defining contraction inequality for a monotone metric is homogeneous under multiplication by a positive constant. Since $g_{\mathrm{BH}}$ is monotone and
\begin{equation*}
g(t)=(1-4t)g_{\mathrm{BH}}
\end{equation*}
with $1-4t>0$ for $t<1/4$, every $g(t)$ is monotone. At $t=1/4$ the metric degenerates to the zero quadratic form. Thus the solution does not exit the monotone cone through a non-monotone metric; it reaches the apex of the cone.
\end{proof}

\begin{remark}[The extinction limit]
The extinction is uniform. All distances are multiplied by $\sqrt{1-4t}$, the diameter tends to zero, and the scalar curvature diverges like $(1-4t)^{-1}$. No state is selected in the limit. In particular, this intrinsic Ricci-flow collapse should not be confused with the behavior of a dissipative quantum channel, such as the depolarizing semigroup, which drives states toward the maximally mixed state. Ricci flow collapses the distinguishability geometry itself.
\end{remark}

\begin{figure}[H]
\centering
\begin{tikzpicture}[x=3.45cm,y=2.45cm,>=Latex]
\draw[->] (0,0) -- (1.78,0) node[right]{\footnotesize $s$};
\draw[->] (0,0) -- (0,1.18) node[above]{\footnotesize $\Phi$};
\draw[very thick] plot coordinates {(0,0) (0.205,0.203) (0.410,0.398) (0.615,0.577) (0.820,0.731) (1.024,0.854) (1.229,0.942) (1.434,0.991) (1.571,1.000)};
\draw[thick] plot coordinates {(0,0) (0.197,0.195) (0.394,0.375) (0.591,0.527) (0.788,0.640) (0.985,0.706) (1.133,0.721)};
\draw plot coordinates {(0,0) (0.153,0.150) (0.305,0.282) (0.458,0.382) (0.611,0.438) (0.702,0.447)};
\draw[densely dashed] plot coordinates {(0,0) (0.082,0.080) (0.164,0.146) (0.246,0.188) (0.314,0.200)};
\node[anchor=west,font=\footnotesize] at (1.45,1.06) {$t=0$};
\node[anchor=west,font=\footnotesize] at (1.14,0.74) {$t=0.12$};
\node[anchor=west,font=\footnotesize] at (0.71,0.46) {$t=0.20$};
\node[anchor=west,font=\footnotesize] at (0.32,0.21) {$t=0.24$};
\node[font=\footnotesize] at (0.85,-0.22) {(a) warping profile};
\end{tikzpicture}
\hfill
\begin{tikzpicture}[x=21.5cm,y=0.72cm,>=Latex]
\draw[->] (0,0) -- (0.285,0) node[right]{\footnotesize $t$};
\draw[->] (0,0) -- (0,3.4);
\draw[very thick] (0,1) -- (0.25,0);
\node[anchor=west,font=\footnotesize] at (0.13,0.62) {$a^2=1-4t$};
\draw[thick,densely dashed] plot coordinates {(0,1) (0.0435,1.21) (0.0870,1.53) (0.1131,1.83) (0.1306,2.09) (0.1480,2.45) (0.1567,2.68) (0.1654,2.95)};
\node[anchor=west,font=\footnotesize] at (0.02,2.6) {$R/6=1/(1-4t)$};
\draw[dotted] (0.25,0) -- (0.25,3.3);
\node[anchor=north,font=\footnotesize] at (0.25,-0.05) {$T=\tfrac14$};
\node[font=\footnotesize] at (0.125,-0.65) {(b) scale and curvature};
\end{tikzpicture}
\caption{Ricci flow of the Bures--Helstrom metric. The warping profile contracts homothetically,
while the scalar curvature diverges at the extinction time $T=1/4$.}
\label{fig:shrinker}
\end{figure}
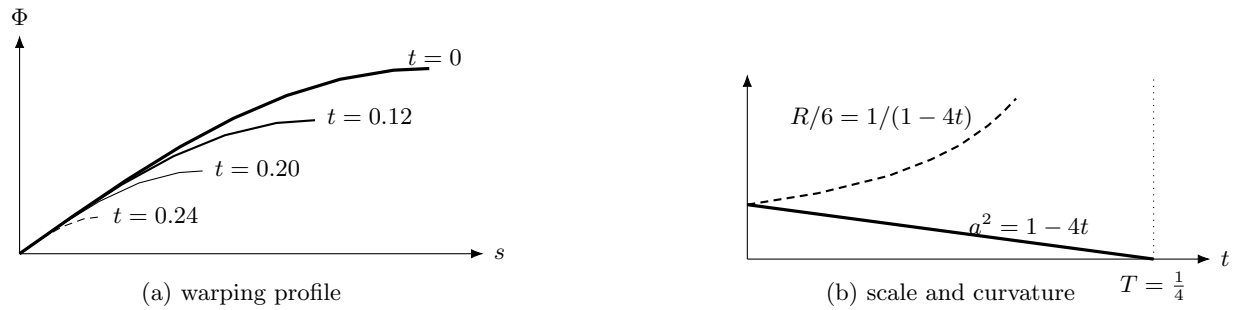

\section{Volume-normalized flow}
\label{sec:normalized}

The unnormalized flow collapses because of the overall scaling degree of freedom. The
volume-normalized Ricci flow removes this scale,
\begin{align}
\label{eq:vol-normalized}
\d_t g&=-2\Ric(g)+\frac{2}{n}\,\bar R\,g,\\
\label{eq:av_curv}
\bar R&=\frac{\int_M R\,dV}{\int_M dV}\,,
\end{align}
where $n=\dim M=3$, and $\bar R$ is the average scalar curvature. This is the flow Hamilton used to deform a positively curved three--metric to a round space form \cite{H82,CK04}, and the Bures--Helstrom metric is one of its fixed points.

We work in the same Hamilton--DeTurck gauge as Section~\ref{sec:gauge}. At the fixed point the scale is stationary, so the comoving label $\xi$ coincides with arclength and the time-dependent diffusivity of \eqref{eq:comoving-heat} reduces to a constant; the scalar flow is then the autonomous equation
\begin{equation}
\label{eq:normalized-flow}
\Psi_t=\Psi_{\xi\xi}-2+\lambda(t)\Psi,
\end{equation}
with multiplier $\lambda=\tfrac{2}{3}\bar R$ inherited from \eqref{eq:vol-normalized}. Equivalently, preserving $\int_0^{\pi/2}\Psi\,d\xi$ requires
\begin{equation}
\label{eq:lambda}
\lambda(t)=\frac{\pi+\Psi_\xi(0,t)-\Psi_\xi(\tfrac{\pi}{2},t)}{\int_0^{\pi/2}\Psi(\xi,t)\,d\xi}\,,
\end{equation}
which at a regular pole, $\Psi_\xi(0,t)=0$, and a reflecting pure--state edge, $\Psi_\xi(\tfrac\pi2,t)=0$, reduces to $\lambda=\pi/\!\int\Psi$. For the Bures--Helstrom profile $\Psi_{\mathrm{BH}}=\sin^2\xi$ one has $\int_0^{\pi/2}\sin^2\xi\,d\xi=\pi/4$ and $\bar R=R_{\mathrm{BH}}=6$, so
\begin{align*}
\lambda_{\mathrm{BH}}&=4\\
&=\frac23\, R_{\mathrm{BH}} .
\end{align*}
Since
\begin{align*}
(\sin^2\xi)_{\xi\xi}-2+4\sin^2\xi&=2\cos2\xi-2+4\sin^2\xi\\
&=0,
\end{align*}
the Bures--Helstrom profile is indeed a fixed point.

\begin{proposition}[Normalized fixed point]
\label{prop:normalized-fixed}
Under the volume-normalized flow \eqref{eq:vol-normalized}, the Bures--Helstrom metric is a fixed point, with $\lambda_{\mathrm{BH}}=4=\tfrac23 R_{\mathrm{BH}}$ in the normalization \eqref{eq:BH-bloch}.
\end{proposition}

\subsection{Linearization}

Since $g_{\mathrm{BH}}$ extends by reflection to the unit round $\bS^3$ (Section~\ref{sec:BH}), we linearize the volume-normalized flow there, using the zonal harmonics of $\bS^3$. The pure states form the totally geodesic equator $\xi=\pi/2$, and the reflecting boundary condition $\Psi_\xi(\pi/2,t)=0$ imposed above retains exactly the rotationally symmetric perturbations that are even under the equatorial reflection $\xi\mapsto\pi-\xi$; the odd perturbations move the pure-state boundary and are not admissible.

Write a rotationally symmetric perturbation through the warping factor,
\begin{equation*}
\Phi=\sin\xi+\epsilon\,w+O(\epsilon^2),
\end{equation*}
with $w$ a zonal function on $\bS^3$. Linearizing the scalar flow \eqref{eq:normalized-flow} about the fixed profile -- the equation in which the DeTurck field vanishes at the fixed point -- gives, modulo the Lagrange multiplier that fixes the volume,
\begin{equation}
\label{eq:A-operator}
\begin{split}
w_t&=\cA\,w,\\
\cA&=\d_{\xi\xi}+2\cot\xi\,\d_\xi+3\\
&=\Delta_{\bS^3}+3,
\end{split}
\end{equation}
the shifted round--sphere Laplacian on rotationally symmetric functions. In the $\Psi$ variable, with $v=\delta\Psi=2\sin\xi\,w$, the operator is conjugate to $\d_{\xi\xi}+4$, and the eigenvalue problem becomes a self-adjoint Sturm--Liouville problem on $[0,\pi/2]$, with boundary conditions $v(0)=0$, at the maximally mixed state, and $v_\xi(\pi/2)=0$, the reflecting pure-state condition. The zonal spherical harmonics on $\bS^3$ have Laplacian eigenvalues $-\ell(\ell+2)$ \cite{P16}, so
\begin{equation}
\label{eq:spectrum}
\sigma_\ell=3-\ell(\ell+2)=-(\ell-1)(\ell+3),
\qquad \ell=0,1,2,\ldots,
\end{equation}
with eigenfunctions $w_\ell=\sin((\ell+1)\xi)/\sin\xi$, equivalently $v_\ell=2\sin((\ell+1)\xi)$. Since $w_\ell(\pi-\xi)=(-1)^\ell w_\ell(\xi)$, the reflecting condition $w'_\ell(\pi/2)=0$ retains exactly the even modes $\ell=0,2,4,\dots$; the odd modes fail it and are excluded.

\begin{theorem}[Linear stability of the normalized Bures--Helstrom metric]
\label{thm:stability}
The Bures--Helstrom metric is a linearly stable fixed point of the volume-normalized Ricci flow, modulo scaling. On the reflection-even perturbations admitted by the pure-state boundary, the rotationally symmetric linearized spectrum is
\begin{equation*}
\sigma_\ell=-(\ell-1)(\ell+3),\qquad \ell=0,2,4,\ldots,
\end{equation*}
so
\begin{equation*}
\sigma_0=3,\qquad \sigma_2=-5,\qquad \sigma_4=-21,\qquad \sigma_6=-45,\ldots .
\end{equation*}
The only non-negative eigenvalue is $\sigma_0=3$, the overall scale, which is fixed by the volume normalization. The genuine deformations have $\ell\ge2$, with $\sigma_\ell\le\sigma_2=-5$, so the fixed point is stable with spectral gap $5$.
\end{theorem}

\begin{proof}
The eigenvalues follow from \eqref{eq:spectrum} together with the restriction to even $\ell$. The constant mode $\ell=0$ rescales the round metric and is fixed by the volume constraint. For even $\ell\ge2$ one has $\sigma_\ell=-(\ell-1)(\ell+3)\le-5$, and these are the genuine geometric deformations. The odd modes are not admissible: in particular the first nonconstant zonal harmonic $w_1=\cos\xi$ satisfies the Obata equation $\nabla^2 w_1=-w_1\,g$ on the unit $\bS^3$, so $\nabla w_1$ is a conformal vector field and the associated perturbation is a Lie derivative of the metric---a pure diffeomorphism; but $w_1$ is odd under $\xi\mapsto\pi-\xi$ and displaces the pure-state equator, so it does not preserve the totally geodesic boundary and is removed by the reflecting condition. The admissible spectrum therefore carries no zero mode, and the fixed point is linearly stable modulo scale, with gap $|\sigma_2|=5$.
\end{proof}

\begin{remark}
The unnormalized and normalized pictures are complementary. In the unnormalized flow the scale mode ($\ell=0$) drives the finite-time extinction $g(t)=(1-4t)g_{\mathrm{BH}}$; the normalized flow fixes that mode and leaves the Bures--Helstrom metric as an asymptotically stable fixed point.
\end{remark}

\section{Outlook}
\label{sec:outlook}

The Bures--Helstrom metric is the constant-curvature endpoint of the monotone qubit metrics, and its Ricci flow is correspondingly explicit. The unnormalized flow remains on the Bures--Helstrom ray and collapses at finite time. The normalized flow fixes the scale and gives a stable round-hemisphere fixed point with a completely computable spectrum.

For general monotone metrics on the qubit state space, the situation is not expected to be homothetic. The radial profile is no longer $\sin\tau$, the curvature is not constant, and the Ricci flow may move nontrivially inside the monotone cone. In that setting one must distinguish three questions: short-time existence of the geometric flow, preservation of Petz monotonicity, and possible finite-time exit from the cone of nondegenerate monotone metrics. The Bures--Helstrom case does not exhibit exit through a non-monotone interior profile; instead, it reaches the apex of the cone by uniform collapse. It is therefore best viewed as the exactly solvable model and normalization point for the broader Ricci-flow problem on the monotone cone, taken up in \cite{sequel}.

\medskip\noindent
\textbf{Declaration of generative AI and AI-assisted technologies in the manuscript preparation process}: During the preparation of this work, the author used ChatGPT (OpenAI) and Claude (Anthropic) for language editing and stylistic refinement of portions of the manuscript. All substantive ideas, models, derivations, and conclusions are solely those of the author. After using these tools, the author reviewed and edited the content as needed and takes full responsibility for the content of the published article.

\end{document}